\documentclass[12pt,reqno]{amsart} 

\usepackage{mathdots}
\usepackage{amsmath, amssymb} 
\usepackage{amscd}            
\usepackage{amsxtra}          
\usepackage{upref}   
\usepackage{enumerate}         

\usepackage[mathscr]{eucal}

\theoremstyle{plain}

\newtheorem{theorem}{Theorem}

\newtheorem{lemma}[theorem]{Lemma}

\newtheorem{definition}[theorem]{Definition}
\newtheorem{example}[theorem]{Example}

\theoremstyle{remark}

\numberwithin{equation}{section}
\numberwithin{theorem}{section}
\newcommand{\be}%
  {\protect\setcounter{equation}{\value{subsubsection}}}  
\newcommand{\ee}%
  {\protect\setcounter{subsubsection}{\value{equation}}}

\newcommand{\Z}{\mathbb{Z}}

\newcommand{\xn}{x^n - 1}

\begin{document}


\title {On cyclic codes over the ring $ \Z_p + u\Z_p +\cdots + u^{k-1}\Z_p $}
\author{Abhay Kumar Singh and Pramod Kumar Kewat} 
\address{Department of Applied Mathematics\\ 
         Indian School of Mines\\
         Dhanbad 826 004,  India}
\email{singh.ak.am@ismdhanbad.ac.in\\ kewat.pk.am@ismdhanbad.ac.in}

\begin{abstract}
 In this paper, we study cyclic codes over the ring $ \Z_p + u\Z_p +\cdots + u^{k-1}\Z_p $, where $u^k =0$. We find a set of generator for these codes. We also study the rank, the dual and the Hamming distance of these codes. 
\end{abstract}

\maketitle


\markboth{A.K. Singh and P.K. Kewat}{On cyclic codes over the ring $R_k$}

\section{Introduction}
Let $R$ be a ring. A linear code of length $n$ over $ R$ is a $R$ submodule of $R^{n}$. A linear code $C$ of length $n$ over $R$ is cyclic if $ (c_{n-1}, c_1, \dots, c_0) \in C$ whenever $ (c_0, c_1, \dots, c_{n-1}) \in C$. We can consider a cyclic code $C$ of length $n$ over $R$ as an ideal in $R[x]/<x^n - 1>$ via the following correspondence
$$ R^{n} \longrightarrow R[x]/<x^n - 1>, ~~ (c_0, c_1, \dots, c_{n-1})\mapsto c_0 + c_1x + \cdots + c_{n-1}x^{n-1}.$$
In recent time, cyclic codes over rings have been studied extensively because of their important role in algebraic coding theory. The structure of cyclic codes of odd length over rings has been discussed in a series of papers \cite{Bonn-Udaya99, Cal-Slo98, Dou-Shiro01, Lint91}. In \cite{Cal-Slo95}, \cite{Con-Slo93} and \cite{Ples-Qian96}, a complete structure of cyclic codes of odd length over $\Z_4$ has been presented. In \cite{Blac03}, Blackford studied cyclic codes of length $ n = 2k$, when $k$ is odd. The cyclic codes of length a power of $2$ over $ \Z_4$ are studied in \cite{Tah-Oeh04, Tah-Oeh03}. Bonnecaze and Udaya in \cite{Bonn-Udaya99} studied cyclic codes of odd length over $ R_2 = \Z_2 + u \Z_2, u^2 = 0$. In \cite{Tah-Siap07}, Abualrub and Siap studied cyclic codes of an arbitrary length over $ R_2 = \Z_2 + u \Z_2, u^2 = 0$ and over $ R_3 = \Z_2 + u \Z_2 + u^2 \Z_2, u^3 = 0$. Al-Ashker and Hamoudeh in \cite{Ash-Ham11} extended some of the results in \cite{Tah-Siap07} to the ring $R_k = \Z_2 + u\Z_2 + \cdots + u^{k-1}\Z_2$, $u^k = 0.$\\
\indent
Let $R_k = \Z_p + u\Z_p + \cdots + u^{k-1}\Z_p$ where $p$ is a prime number and $u^k = 0.$ In this paper, we discuss the structure of cyclic codes of arbitrary length over the ring $R_k$. We find a set of generators and a minimal spanning set for these codes. We also discuss about the rank and the Hamming distance of these codes. Recall that the Hamming weight of a codeword $c$ is defined as the number of non-zero enteries of $c$ and the Hamming distance of a code $C$ is the smallest possible weight among all its non zero codewords.\\
\indent
Let $C$ be a cyclic code over the ring $R_k = \Z_p + u\Z_p + \cdots + u^{k-1}\Z_p$ where $u^k = 0.$ The line of arguments we have used to find a set of genertors and a minimal spanning set of a code $C$ are somewhat similar to those discussed in \cite{Tah-Siap07, Ash-Ham11}. Note that some slight modification needed in our case in order to find a set of genrator, e.g., the proofs of Lemmas \ref{lemma-c3-main} and \ref{lemma-c4-main} are slightly different from those discussed in \cite{Tah-Siap07, Ash-Ham11} where the proof is not very clear. Again, the line of arguments we have used to find minimum distance are similar to \cite{Tah-Siap07} but slightly different.\\
\indent
The paper is organized as follows. In Section 2, we give a set of generators for the cyclic codes $C$ over the ring $R_k = \Z_p + u\Z_p + \cdots + u^{k-1}\Z_p$ where $u^k = 0.$ In Section 3, we find minimal spanning sets for these codes and dicuss about the rank. In Section 4, we find the minimum distance of these codes. In Section 5, we discuss some of the examples of these codes.
\section{A generator for cyclic codes over the ring $R_k$}
Let $R_k = \Z_p + u_k\Z_p + \cdots + u_{k}^{k-1}\Z_p,~ u_{k}^k = 0.$ A cyclic code $C$ of length $n$ over $R_k$ can be considered as an ideal in the $R_{k,n} = R_k[x]/<x^n - 1>$. Let $C_k$ be a cyclic code of length $n$ over $R_k$. We also consider $C_k$ as an ideal in $R_{k,n}$. We define the map $\psi_{k-1} : R_k \longrightarrow R_{k-1}$ by $\psi_{k-1}( a_0+u_ka_1+ \cdots +u_{k}^{k-1}a_{k-1}) = a_0+u_{k-1}a_1+ \cdots +u_{k-1}^{k-2}a_{k-2}$, where $a_i \in \Z_p$. The map $\psi_{k-1}$ is a ring homomorphism. We extend it to a homomorphism $\phi_{k-1} : C_k \longrightarrow R_{k-1,n}$ defined by
$$ \phi_{k-1}(c_0 + c_1x + \cdots + c_{n-1}x^{n-1}) = \psi_{k-1}(c_0) + \psi_{k-1}(c_1)x + \cdots + \psi_{k-1}(c_{n-1})x^{n-1}.$$

Let $J_{k-1} = \{r(x) \in \Z_p[x] : u_{k}^{k-1}r(x) \in \text{ker} \phi_{k-1}\}.$ It is easy to see that $J_{k-1}$ is an ideal in $R_{1,n}$. Since $R_{1,n}$ is a principal ideal ring, we have $J_{k-1} = <a_{k-1}(x)>$ and $\text{ker}\phi_{k-1} = < u_{k}^{k-1}a_{k-1}(x)>$ with $a_{k-1}(x)|(x^n-1)$ mod $p$. 

Let $C_{k-2}$ be a cyclic code of length $n$ over $R_{k-2}$. We define the map $\psi_{k-2} : R_{k-1} \longrightarrow R_{k-2}$ by $\psi_{k-2}( a_0+u_{k-1}a_1+ \cdots +u_{k-1}^{k-2}a_{k-2}) = a_0+u_{k-2}a_1+ \cdots +u_{k-2}^{k-3}a_{k-3}$, where $a_i \in \Z_p$. The map $\psi_{k-2}$ is a ring homomorphism. We extend it to a homomorphism $\phi_{k-2} : C_{k-1} \longrightarrow R_{k-2,n}$ defined by
$$ \phi_{k-2}(c_0 + c_1x + \cdots + c_{n-1}x^{n-1}) = \psi_{k-2}(c_0) + \psi_{k-2}(c_1)x + \cdots + \psi_{k-2}(c_{n-1})x^{n-1}.$$ 

Let $J_{k-2} = \{r(x) \in \Z_p[x] : u_{k-1}^{k-2}r(x) \in \text{ker} \phi_{k-2}\}.$ We see that $J_{k-2}$ is an ideal in $R_{1,n}$. As above, we have $J_{k-2} = <a_{k-2}(x)>$ and $\text{ker}\phi_{k-2} = < u_{k-1}^{k-2}a_{k-2}(x)>$ with $a_{k-2}(x)|(x^n-1)$ mod $p$.

We continue in the same way as above and define $\psi_{k-3}$, $\psi_{k-4}, \cdots, \psi_2$ and $\phi_{k-3}, \phi_{k-4}, \cdots, \phi_2$. We define $\psi_1 : R_{2} \longrightarrow R_1 = \Z_p $ by $\psi_{k}(a_0 + u_2a_1) = a_{0}.$ The map $\psi_1$ is a ring homomorphism. We extend $\psi_1$ to a homomorphism $\phi_1 : C_2 \longrightarrow R_{1,n}$ defined by
$$ \phi_1(c_0 + c_1x + \cdots + c_{n-1}x^{n-1}) = \psi_1(c_0) + \psi_1(c_1)x + \cdots + \psi_1(c_{n-1})x^{n-1}.$$ 

As above, we have $\text{ker} \phi_1 = <u_2 a_1(x)>$ with $a_1(x)|(x^n - 1)$ mod $p$. The image of $\phi_1$ is an ideal in $R_{1,n}$ and hence a cyclic code in $\Z_p$. Since $R_{1,n}$ is a principal ideal ring, the image of $\phi_1$ is generated by some $g(x) \in \Z_p[x]$ with $g(x)|(x^n - 1)$. Hence, we have  $C_2 = <g(x) + u_2 p(x), u_2 a_1(x)> $ for some $p(x) \in \Z_p[x]$. We have
$$ \phi_1(\dfrac{x^n - 1}{g(x)}(g(x) + u_2 p(x)) = \phi_1(u_2 p(x) \dfrac{x^n - 1}{g(x)}) = 0.$$
Therefore, $u_2 p(x)\dfrac{x^n - 1}{g(x)} \in \text{ker} \phi_1 = < u_2 a_1(x)>$. Hence, $a_1(x)| p(x)\dfrac{x^n - 1}{g(x)}$. Also we have $u_2 g(x) \in \text{ker}\phi_1.$ This implies that $ a_{1}(x) | g(x)$. 

\begin{lemma} \label{lemma-c2}
Let $C_2$ be a cyclic code over $R_2 = \Z_p + u \Z_p, u^2 = 0$. If $C_2 = <g(x) + u p(x), u a_1(x)> $, and $g(x) = a_1(x)$ with ${\rm deg}~g(x) = r$, then 
$$ C_2 = <g(x) + u p(x)> ~\text{and}~ (g(x) + u p(x))|(x^n - 1) ~\text{in}~ R_2.$$ 
\end{lemma}
\begin{proof}
We have $ u (g(x) + u p(x)) = u g(x)$ and $g(x) = a(x)$. It is clear that $ C_2 \subset < g(x) + u p(x)>.$ Hence, $C_2 = <g(x) + u p(x)>$. By the division algorithm, we have
$$ x^n - 1 = (g(x) + u p(x)) q(x) + r(x), ~~~~ \text{where}~ r(x) = 0~ \text{or deg}~r(x) < r.$$
Since $r(x) \in C_2$, we have $r(x) = 0$ and hence $(g(x) + u p(x))|(x^n - 1) ~\text{in}~ R_2$.
\end{proof}
Note that the image of $\phi_2$ is an ideal in $R_{2,n}$, hence a cyclic code over $R_2$. Therefore, we have $\text{Im}(\phi_2) = <g(x) + u_2 p_1(x), u_2 a_1(x)>$ with $a_1(x) | g(x) | (x^n - 1)$ and $a_1(x) | p_1(x)(\dfrac{x^n - 1}{g(x)}).$ Also, we have $\text{ker}\phi_2 = <u_{3}^2 a_2(x)>$ with $a_2(x)|(x^n - 1)$ mod $p$ and $u_{3}^2 a_1(x) \in \text{ker} \phi_2$. As above, the cyclic code $C_3$ over $R_3$ is given by
$$ C_3 = < g + u_{3} p_1(x) + u_{3}^2 p_{2}(x), u_3 a_1(x) + u_{3}^2 q_1(x), u_{3}^2 a_2(x)> $$
with $a_2(x) | a_1(x) | g(x) | (x^n - 1)$, $a_1(x) | p_1(x)(\dfrac{x^n - 1}{g(x)})$ mod $p$, $ a_2(x) | q_1(x) (\dfrac{x^n - 1}{a_1(x)})$, $a_2(x)|p_1(x) \dfrac{x^n - 1}{g(x)}$ and $ a_2(x)|p_2(x) (\dfrac{x^n - 1}{g(x)})(\dfrac{x^n - 1}{a_1(x)})$. We may assume that deg $p_2(x) < \text{deg} a_2(x)$, $ \text{deg} q_1(x) < \text{deg} a_2(x)$ and deg $p_1(x) < \text{deg} a_1(x)$ because g.c.d.$(a, b)$ = g.c.d.$( a, b + da)$ for any $d$.  We have the following lemma.
\begin{lemma} \label{lemma-c3}
 Let $C_3$ be a cyclic code over $ R_3 = \Z_p + u \Z_p + u^2 \Z_p, u^3 = 0$. If $ C_3 = < g + u p_1(x) + u^2 p_{2}(x), u a_1(x) + u^2 q_1(x), u^2 a_2(x)>$, and $ a_2(x) = g(x)$, then $ C_3 = < g + u p_1(x) + u^2 p_{2}(x)>$ and $(g + u p_1(x) + u^2 p_{2}(x)) | (x^n -1)$ in $R_3$.
\end{lemma}
\begin{proof}
Since $ a_2(x) = g(x)$, we have $ a_1(x) = a_2(x) = g(x)$. From Lemma \ref{lemma-c2}, we get $(g(x) + u p(x) | (\xn)$ in $R_2$, and $ C_3 = < g + u p_1(x) + u^2 p_{2}(x), u^2 a_2(x)>$. The rest of the proof is similar to Lemma \ref{lemma-c2}.
\end{proof}
\begin{lemma} \label{lemma-c3-main}
Let $C_3$ be a cyclic code over $ R_3 = \Z_p + u \Z_p + u^2 \Z_p, u^3 = 0$. If $n$ is relatively prime to $p$, then $ C_3 = <g(x), u a_1(x), u^2 a_2(x)>$ = $<g(x) + u a_1(x) + u^2 a_2(x)> $ over $R_3$.
\end{lemma}
\begin{proof}
Since $n$ is relatively prime to $p$, the polynomial $ x^n - 1$ factors uniquely into a product of distinct irreducible polynomials. This gives, 
$$ \text{g.c.d.}(a_1(x), \dfrac{(x^n - 1)}{g(x)}) = \text{g.c.d.}( a_2(x), \dfrac{(\xn)}{a_1(x)}) = \text{g.c.d.}( a_2(x), \dfrac{(\xn)}{g(x)}) = 1.$$
Since $a_1(x) | p_1(x)(\dfrac{x^n - 1}{g(x)})$, we get $a_1(x) | p_1(x)$. But $\text{deg} p_1(x) < \text{deg} a_1(x),$ hence $p_1(x) = 0.$ We have $a_2(x) | q_1(x)(\dfrac{x^n - 1}{a_1(x)})$ and $a_2(x) | p_2(x)(\dfrac{x^n - 1}{g(x)}) (\dfrac{x^n - 1}{a_1(x)})$, this gives $a_2(x) | q_1(x)$ and $a_2(x) | p_2(x)$. But $\text{deg} q_1(x) < \text{deg} a_2(x)$ and $\text{deg} p_2(x) < \text{deg} a_2(x),$ hence $p_2(x) = q_1(x) = 0.$ So, $C_3 = <g(x), u a_1(x), u^2 a_2(x)>$. Let $h(x) = g(x) + u a_1(x) + u^2 a_2(x).$ Then
$$ u^2 h(x) = u^2 g(x),~\dfrac{\xn}{a_1(x)}h(x) = \dfrac{\xn}{a_1(x)}u^2 a_2(x),~\text{and}$$
$$u\dfrac{\xn}{g(x)}h(x) = \dfrac{\xn}{g(x)}u^2 a_1(x) \in <h(x)>.$$
Since $n$ is relatively prime to $p$, we have
$$ \text{g.c.d.}\left(g(x), \dfrac{(x^n - 1)}{g(x)}\right) = \text{g.c.d.}\left(a_1(x), \dfrac{(x^n - 1)}{a_(x)}\right) = 1.$$
Hence, $ 1 = f_1(x) \dfrac{(x^n - 1)}{g(x)} + f_2(x) g(x)$, for some polynomial $f_1(x)$ and $f_2(x)$, and $ 1 = m_1(x) \dfrac{(x^n - 1)}{a_1(x)} + m_2(x) a_1(x)$, for some polynomial $m_1(x)$ and $m_2(x)$. Therefore, $ u^2 a_1(x) = u^2 a_1(x)f_1(x) \dfrac{(x^n - 1)}{g(x)} + u^2 a_1(x)f_2(x) g(x) \in <h(x)>$, $ u^2 a_2(x) =u^2 a_2(x) m_1(x) \dfrac{(x^n - 1)}{a_1(x)} + u^2 a_2(x)m_2(x) a_1(x) \in <h(x)>$  and hence $g(x) + ua_1(x) \in < h(x) >.$ We have $ (g(x) + ua_1(x))^2 = g(x)^2 + 2ug(x)a_1(x) + u^2 a_1(x)^2.$ Since $ u^2 a_1(x)^2 \in <h(x)>$, we have $ g(x)^2 + 2ua_1(x)g(x) \in <h(x)>$ and $ug(x)^2 + 2u^2a_1(x)g(x) \in <h(x)>$. So, $u g(x)^2 \in <h(x)>$. We have $ u g(x) = u f_2(x) g(x)^2.$ Hence, $u g(x) \in <h(x)>$. We have $$ \dfrac{\xn}{g(x)}h(x) = \dfrac{\xn}{g(x)}u a_1(x) + \dfrac{\xn}{g(x)}u^2 a_2(x).$$ Since $u^2 a_2(x) \in <h(x)>,$ this gives $ \dfrac{\xn}{g(x)}u a_1(x) \in <h(x)>$. We also have
$$ u a_1(x) = f_1(x) \dfrac{(x^n - 1)}{g(x)} u a_1(x) + f_2(x) ug(x) a_1(x). $$
This gives, $u a_1(x) \in <h(x)>$ and hence $ g(x) \in <h(x)>$. Therefore, $ C_3$ = $<g(x), u a_1(x)$, $u^2 a_2(x)>$ = $<g(x) + u a_1(x) + u^2 a_2(x)> $.
\end{proof}
Note that the image of $\phi_3$ is an ideal in $R_{3,n}$, hence a cyclic code over $R_3$. Therefore, we have $\text{Im}(\phi_3) = <g + u_{3} p_1(x) + u_{3}^2 p_{2}(x), u_3 a_1(x) + u_{3}^2 q_1(x), u_{3}^2 a_2(x)>$ with $a_2(x) | a_1(x) | g(x) | (x^n - 1)$ and $a_1(x) | p_1(x)\left(\frac{x^n - 1}{g(x)}\right),$ $a_2(x) | q_1(x)\left(\frac{x^n - 1}{a_2(x)}\right)$ and $a_1(x) | p_2(x)\left(\frac{x^n - 1}{g(x)}\right)\left(\frac{x^n - 1}{a_1(x)}\right).$ Also, we have $\text{ker}\phi_3 = <u_{4}^3 a_3(x)>$ with $a_3(x)|(x^n - 1)$ mod $p$ and $u_{4}^3 a_2(x) \in \text{ker} \phi_3$. As above, the cyclic code $C_4$ over $R_4$ is given by
$ C_4 = < g + u_{4} p_1(x) + u_{4}^2 p_{2}(x) + u_{4}^3 p_3(x), u_4 a_1(x) + u_{4}^2 q_1(x) + u_{4}^3 q_2(x), u_{4}^2 a_2(x) + u_{4}^3 l_{1}(x), u_{4}^3 a_{3}(x)> $
with $a_{3}(x) | a_2(x) | a_1(x) | g(x) | (x^n - 1)$, $a_1(x) | p_1(x)\left(\frac{x^n - 1}{g(x)}\right)$ mod $p$, $ a_2(x) | q_1(x) \left(\frac{x^n - 1}{a_1(x)}\right)$, $a_2(x)|p_1(x) \left(\frac{x^n - 1}{g(x)}\right)$, $ a_2(x)|p_2(x) \left(\frac{x^n - 1}{g(x)}\right)\left(\frac{x^n - 1}{a_1(x)}\right)$,\\ $ a_3(x) | l_1(x) \left(\frac{x^n - 1}{a_2(x)}\right)$, $ a_3(x) | q_2(x) \left(\frac{x^n - 1}{q_1(x)}\right) \left(\frac{x^n - 1}{a_1(x)}\right)$ and $ a_3(x) | p_3(x) \left(\frac{x^n - 1}{g(x)}\right) \left(\frac{x^n - 1}{a_2(x)}\right) \times \left(\frac{x^n - 1}{a_1(x)}\right). $ We may assume that $ \text{deg} p_3(x) < \text{deg} a_3(x)$, $ \text{deg} q_2(x) < \text{deg} a_3(x)$, $ \text{deg} l_1(x) < \text{deg} a_3(x)$, deg $p_2(x) < \text{deg} a_2(x)$, $ \text{deg} q_1(x) < \text{deg} a_2(x)$ and deg $p_1(x) < \text{deg} a_1(x)$ because g.c.d.$(a, b)$ = g.c.d.$( a, b + da)$ for any $d$.  We have the following lemma.
\begin{lemma} \label{lemma-c4}
 Let $C_4$ be a cyclic code over $ R_4 = \Z_p + u \Z_p + u^2 \Z_p + u^3 \Z_p, u^4 = 0$. If $C_4 = < g + u p_1(x) + u^2 p_{2}(x) + u^3 p_3(x), u a_1(x) + u^2 q_1(x) + u^3 q_2(x), u^2 a_2(x) + u^3 l_{1}(x), u^3 a_{3}(x)>$, and $ a_3(x) = g(x)$, then $ C_4 = < g + u p_1(x) + u^2 p_{2}(x) + u^3 p_3(x)>$ and $(g + u p_1(x) + u^2 p_{2}(x) + u^3 p_3(x)) | (x^n -1)$ in $R_4$.
\end{lemma}
\begin{proof}
Since $a_3(x) = g(x)$, we have $ a_1(x) = a_2(x) = a_3(x) = g(x)$. From Lemma \ref{lemma-c3}, we get $(g(x) + u p_1(x) +u^2 p_2(x)) | (\xn)$ in $R_3$, and $ C_4 = <g + u p_1(x) + u^2 p_{2}(x) + u^3 p_3(x), u a_1(x) + u^2 q_1(x) + u^3 q_2(x), u^3 a_{3}(x)>$. The rest of the proof is similar to Lemma \ref{lemma-c3}.
\end{proof}
\begin{lemma} \label{lemma-c4-main}
Let $C_4$ be a cyclic code over $ R_4 = \Z_p + u \Z_p + u^2 \Z_p + u^3 \Z_p, u^4 = 0$. If $n$ is relatively prime to $p$, then $ C_4 = <g(x), u a_1(x), u^2 a_2(x), u^3 a_3(x)>$ = $<g(x) + u a_1(x) + u^2 a_2(x) + u^3 a_3(x)> $ over $R_4$.
\end{lemma}
\begin{proof}
The proof is similar to Lemma \ref{lemma-c3-main}. Since $n$ is relatively prime to $p$, the polynomial $ x^n - 1$ factors uniquely into a product of distinct irreducible polynomials. This gives,
\begin{center}
$\text{g.c.d.}\left(a_1(x), \frac{(x^n - 1)}{g(x)}\right) = \text{g.c.d.}\left( a_2(x), \frac{(\xn)}{a_1(x)}\right) = \text{g.c.d.}\left( a_2(x), \frac{(\xn)}{g(x)}\right) = 1,$
\end{center}
\begin{center}
 $\text{g.c.d.}\left( a_3(x), \frac{(\xn)}{a_2(x)}\right) = \text{g.c.d.}\left( a_3(x), \frac{(\xn)}{a_1(x)}\right) = \text{g.c.d.}\left( a_3(x), \frac{(\xn)}{g(x)}\right) = 1$.
\end{center}
Since $a_1(x) | p_1(x)\left(\frac{x^n - 1}{g(x)}\right)$, we get $a_1(x) | p_1(x)$. But $\text{deg} p_1(x) < \text{deg} a_1(x),$ hence $p_1(x) = 0.$ We have $a_2(x) | q_1(x)\left(\frac{x^n - 1}{a_1(x)}\right)$ and $a_2(x) | p_2(x)\left(\frac{x^n - 1}{g(x)}\right) \left(\frac{x^n - 1}{a_1(x)}\right)$, this gives $a_2(x) | q_1(x)$ and $a_2(x) | p_2(x)$. But $\text{deg} q_1(x) < \text{deg} a_2(x)$ and $\text{deg} p_2(x) < \text{deg} a_2(x),$ hence $p_2(x) = q_1(x) = 0.$ Similarly, $ p_3(x) = q_2(x) = l_1(x) = 0.$ So, $C_4 = <g(x), u a_1(x), u^2 a_2(x), u^3 a_3(x)>$. Let $h(x) = g(x) + u a_1(x) + u^2 a_2(x) + u^3 a_3(x).$ Then
$$ u^3 h(x) = u^3 g(x),~\dfrac{\xn}{a_2(x)}h(x) = \dfrac{\xn}{a_2(x)}u^3 a_3(x),$$
$$u\dfrac{\xn}{a_1(x)}h(x) = \dfrac{\xn}{a_1(x)}u^3 a_2(x) ~\text{and}~ u^2\dfrac{\xn}{g(x)}h(x) = \dfrac{\xn}{g(x)}u^3 a_1(x) \in <h(x)>.$$
Since $n$ is relatively prime to $p$, we have
\begin{center} $\text{g.c.d.}\left(g(x), \frac{(x^n - 1)}{g(x)}\right) = \text{g.c.d.}\left(a_1(x), \frac{(x^n - 1)}{a_1(x)}\right) = \text{g.c.d.}\left(a_2(x), \frac{(x^n - 1)}{a_2(x)}\right) = 1.$\end{center}
Hence, $ 1 = f_1(x) \frac{(x^n - 1)}{g(x)} + f_2(x) g(x)$, for some polynomials $f_1(x)$ and $f_2(x)$,  $ 1 = m_1(x) \frac{(x^n - 1)}{a_1(x)} + m_2(x) a_1(x)$, for some polynomials $m_1(x)$ and $m_2(x)$ and $ 1 = n_1(x) \frac{(x^n - 1)}{a_2(x)} + n_2(x) a_2(x)$, for some polynomials $n_1(x)$ and $n_2(x)$. Therefore,
\begin{center}
$ u^3 a_1(x) = u^3 a_1(x)f_1(x) \frac{(x^n - 1)}{g(x)} + u^3 a_1(x)f_2(x) g(x) \in <h(x)>$,\\
$ u^3 a_2(x) =u^3 a_2(x) m_1(x) \frac{(x^n - 1)}{a_1(x)} + u^3 a_2(x)m_2(x) a_1(x) \in <h(x)>$ and
$ u^3 a_3(x) =u^3 a_3(x) n_1(x) \frac{(x^n - 1)}{a_2(x)} + u^3 a_3(x)n_2(x) a_2(x) \in <h(x)>.$
\end{center}
Hence, $g(x) + ua_1(x) + u^2a_2(x)\in < h(x)>.$ We have $ (g(x) + ua_1(x) + u^2a_2(x))^2 = g(x)^2 + u^2 a_1(x)^2 + 2ug(x)a_1(x) +  2u^2g(x)a_2(x) + 2u^3a_1(x)a_2(x).$ Since $ u^3 a_2(x)\in <h(x)>$, we have $ g(x)^2 + u^2 a_1(x)^2 + 2ua_1(x)g(x) + 2u^2g(x)a_2(x)$  $\in <h(x)>$ and hence $u^2g(x)^2  \in <h(x)>$. We have $ u^2 g(x) = u^2 f_2(x) g(x)^2.$ Hence, $u^2 g(x) \in <h(x)>$. We have $$ \frac{\xn}{a_1(x)}h(x) = \frac{\xn}{a_1(x)}u^2 a_2(x) + \frac{\xn}{a_1(x)}u^3 a_3(x) ~\text{and}$$
$$u\frac{\xn}{g(x)}h(x) = \frac{\xn}{g(x)}u^2 a_1(x) + \frac{\xn}{g(x)}u^3 a_2(x).$$
This gives, $\dfrac{\xn}{g(x)}u^2 a_1(x) \in <h(x)>$ and $\dfrac{\xn}{a_1(x)}u^2 a_2(x) \in <h(x)>.$
We have
$$ u^2 a_1(x) = f_1(x) \dfrac{(x^n - 1)}{g(x)} u^2 a_1(x) + f_2(x) u^2g(x) a_1(x). $$
Therefore, $u^2 a_1(x) \in <h(x)>$. We also have
$$ u^2 a_2(x) = m_1(x) \frac{(x^n - 1)}{a_1(x)}u^2 a_2(x) + u^2m_2(x) a_1(x)a_2(x).$$
Therefore, $ u^2 a_2(x) \in <h(x)>.$ Hence, $ g(x) + u a_1(x) \in <h(x)>$. The rest of the proof is similar to Lemma \ref{lemma-c3-main}, but for readers convenience we repeat the proof here.
We have $ (g(x) + ua_1(x))^2 = g(x)^2 + 2ug(x)a_1(x) + u^2 a_1(x)^2.$ Since $ u^2 a_1(x) \in <h(x)>$, we have $ g(x)^2 + 2ua_1(x)g(x) \in <h(x)>$ and $ug(x)^2 + 2u^2a_1(x)g(x) \in <h(x)>$. So, $u g(x)^2 \in <h(x)>$. We have $ u g(x) = u f_2(x) g(x)^2.$ Hence, $u g(x) \in <h(x)>$. We have $$ \dfrac{\xn}{g(x)}h(x) = \dfrac{\xn}{g(x)}u a_1(x) + \dfrac{\xn}{g(x)}u^2 a_2(x) + \frac{\xn}{g(x)}u^3 a_3(x).$$ Since $u^2 a_2(x), u^3a_3(x) \in <h(x)>,$ this gives $ \dfrac{\xn}{g(x)}u a_1(x) \in <h(x)>$. We also have
$$ u a_1(x) = f_1(x) \dfrac{(x^n - 1)}{g(x)} u a_1(x) + f_2(x) ug(x) a_1(x). $$
This gives, $u a_1(x) \in <h(x)>$ and hence $ g(x) \in <h(x)>$. Therefore, $ C_4$ = $<g(x), u a_1(x)$, $u^2 a_2(x), u^3 a_3(x)>$ = $<g(x) + u a_1(x) + u^2 a_2(x) + u^3 a_3(x)>. $ This proves the lemma.
\end{proof}
Following the same process as above and by induction on $k$, we get the following theorem.
\begin{theorem} \label{gen-main}
Let $C_k$ be a cyclic code over $ R_k = \Z_p + u \Z_p + u^2 \Z_p + \cdots + + u^{k-1} \Z_p, u^k = 0$.
\begin{enumerate} [{\rm (1)}]
\item If $n$ is relatively prime to $p$, then we have $ C_k = <g(x), u a_1(x), u^2 a_2(x), \dots,\\ u^{k-1} a_{k-1}(x)>$ = $<g(x) + u a_1(x) + u^2 a_2(x) + \dots + u^{k-1} a_{k-1}(x)> $ over $R_k$.
\item If $n$ is not relatively prime to $p$, then
\begin{enumerate}[{\rm ($a$)}]
\item $C_k = < g(x) + u p_1(x) + u^2 p_{2}(x) + \cdots + u^{k-1} p_{k-1}(x)>$ where $g(x)$ and $p_i(x)$ are polynomials in $\Z_p[x]$ for each $i =1,2, \dots,k-1$ with $g(x)|(\xn)$ mod $p$, $(g(x) + u p_1(x) + u^2 p_{2}(x) + \cdots + u^{k-1} p_{k-1}(x))|(\xn)$ in $R_k$ and $ {\rm deg} p_i < {\rm deg} p_{i-1}$ for all $ 1 \leq i \leq k.$ {\rm Or}
\item $C_k = < g(x) + u p_1(x) + u^2 p_{2}(x) + \cdots + u^{k-1} p_{k-1}(x), u^{k-1}a_{k-1}(x)>$ where $a_{k-1}(x)|g(x)|(\xn)$ mod $p$, $g(x) + u p(x)|(x^n - 1)$ in $R_2$, $g(x)|p_1(x)\left(\frac{x^n-1}{g(x)}\right)$ and $a_{k-1}(x)|p_1(x)\left(\frac{x^n-1}{g(x)}\right)$, $a_{k-1}(x)|p_2(x)\left(\frac{x^n-1}{g(x)}\right)\left(\frac{x^n-1}{g(x)}\right),\dots, \\a_{k-1}(x)|p_{k-1}(x) \tiny{\underbrace{\left(\frac{x^n-1}{g(x)}\right) \cdots \left(\frac{x^n-1}{g(x)}\right)}_{k-1~{\rm times}}}$ and ${\rm deg}p_{k-1}(x) < {\rm deg}a_{k-1}(x)$. {\rm Or}
\item $C_k = < g + u p_1(x) + u^2 p_{2}(x) + \cdots + u^{k-1} p_{k-1}(x), u a_1(x) + u^2 q_1(x) + \cdots + u^{k-1} q_{k-2}(x), u^2 a_2(x) + u^3 l_{1}(x) + \cdots + u^{k-1}l_{k-3}(x), \dots, u^{k-2} a_{k-2}(x) + u^{k-1} t_1(x), u^{k-1} a_{k-1}(x)>$ with $a_{k-1}(x)|a_{k-2}(x)| \cdots | a_2(x)|a_1(x)|g(x)|(\xn)$ mod $p$ $ a_{k-2}(x)|p_1(x)\left(\frac{\xn}{g(x)}\right), \dots, a_{k-1}|t_1(x)\left(\frac{\xn}{a_{k-2}(x)}\right)$, $\dots$, $a_{k-1}|p_{k-1} \times\left(\frac{\xn}{g(x)}\right)\cdots\left(\frac{\xn}{a_{k-2}(x)}\right)$. Moreover, ${\rm deg} p_{k-1}(x) < {\rm deg}a_{k-1}(x), \dots, {\rm deg} t_1(x) < a_{k-1}(x), \dots,$ and $ {\rm deg} p_1(x) < {\rm deg}a_{k-2}(x)$. 
\end{enumerate}
\end{enumerate}
\end{theorem}
\section{Ranks and minimal spanning sets}
\begin{theorem} \label{rank2}
Let $n$ is not relatively prime to $p$. Let $C_2$ be a cyclic code of length $n$ over $R_2 = \Z_p + u \Z_p, u^2 = 0.$ 
\begin{enumerate} [{\rm (1)}]
 \item If $ C_2 = <g(x) + u p(x)> $ with deg $g(x) = r$ and $(g(x) + u p(x))|(x^n - 1)$, then $C_2$ is a free module with rank $n-r$ and a basis $B_1 = \{ g(x) + u p(x), x(g(x) + u p(x)), \dots, x^{n-r-1}(g(x) + up(x))\},$ and $|C_2| = p^{2n-2r}.$
\item If $ C_2 = <g(x) + u p(x), ua(x)> $ with deg $g(x) = r$ and deg $a(x) =t$, then $C_2$ has rank $n-t$ and a minimal spanning set $B_2 = \{ g(x) + u p(x), x(g(x) + u p(x)), \dots, x^{n-r-1}(g(x) + up(x)), ua(x), xua(x), \dots, x^{r-t-1}ua(x)\},$ and $|C_2|\\ = p^{2n-r-t}.$
\end{enumerate}
\end{theorem}
\begin{proof}
(1) Suppose $\xn = (g(x) + u p(x))(h(x) + u h_1(x))$ over $R_2$. Let $c(x) \in C_2 = <g(x) + u p(x)> $, then $c(x) = (g(x) + u p(x))f(x)$ for some polynomial $f(x)$. If deg $f(x) \leq n-r-1$, then $c(x)$ can be written as linear combinations of elements of $B_1$. Otherwise by the division algorithm there exist polynomials $q(x)$ and $r(x)$ such that
$$ f(x) = \left(\frac{\xn}{g(x) + u p(x)}\right)q(x) + r(x) ~\text{where}~ r(x) = 0 ~\text{or}~ \text{deg}~r(x) \leq n-r-1. $$
This gives,
\begin{align*}
(g(x) + u p(x))f(x)= & (g(x) + u p(x))\left(\left(\frac{\xn}{g(x) + u p(x)}\right)q(x) + r(x)\right)\\
= & (g(x) + u p(x))r(x).
\end{align*}
Since deg $r(x) \leq n-r-1$, this shows that $B_1$ spans $C_2$. Now we only need to show that $B_1$ is linearly independent. Let $ g(x) = g_0 + g_1x + \cdots + g_rx^r$ and $p(x) = p_0 + p_1x + \cdots + p_lx^l,$ $ g_0 \in \Z_p^{\times}, g_i, p_{i-1} \in \Z_p, i \geq 1$. Suppose
$$(g(x)+up(x))c_0 + x(g(x)+up(x))c_1 + \cdots + x^{n-r-1}(g(x)+up(x))c_{n-r-1} = 0.$$  
By comparing the coefficients in the above equation, we get
$$ (g_0 + u p_0) c_0 = 0. ~\text{(constant coefficient)} $$
Since $(g_0 + u p_0)$ is unit, we get $ c_0 = 0$. Thus,
$$x(g(x)+up(x))c_1 + \cdots + x^{n-r-1}(g(x)+up(x))c_{n-r-1} = 0.$$
Again comparing the coefficients, we get
$$ (g_0 + u p_0) c_1 = 0. ~\text{(coefficient of}~x). $$
As above, this gives $ c_1 = 0$. Continuing in this way we get that $c_i = 0 $ for all $ i = 0, 1. \dots, n-r-1$. Therefore, the set $B_1$ is linearly independent and hence a basis for $C_2$.\\
(2) If $C_2 = <g(x) + u p(x), ua(x)> $ with deg $g(x) = r$ and deg $a(x) =t$. The lowest degree polynomial in $C_2$ is $u a(x)$. It is suffices to show that $B_2$ spans $ B =\{ g(x) + u p(x), x(g(x) + u p(x)), \dots, x^{n-r-1}(g(x) + up(x)), ua(x), xua(x), \dots,\\ x^{n-t-1}ua(x)\}.$ We first show that $u x^{r-t} a(x) \in \text{span}(B_2)$. Let the leading coefficients of $ x^{r-t} a(x)$ be $a_0$ and of $g(x) + u p(x)$ be $g_0$. There exists a constant $c_0 \in \Z_p$ such that $ a_0 = c_0 g_0$. Then we have
$$ u x^{r-t} a(x) =  u c_0(g(x) + u p(x)) + u m(x), $$
where $um(x)$ is a polynomial in $ C_2$ of degree less than $r$. Since $C_2 = <g(x) + u p(x), ua(x)> $, any polynomial in $C_2$ must have degree greater or equal to deg $a(x) = t$. Hence, $ t \leq \text{deg}~m(x) <r$ and
$$ um(x) = \alpha_0 u a(x) + \alpha_1 x u a(x) + \cdots + \alpha_{r-t-1} x^{r-t-1} u a(x).$$
Thus, $u x^{r-t} a(x) \in \text{span}(B_2).$ Inductively, we can show that $u x^{r-t+1} a(x), \dots,\\ ux^{n-t-1}a(x) \in \text{span}(B_2).$ Hence $B_2$ is a generating set. As in (1), by comparing the coefficients we can see that $B_2$ is linearly independent. Therefore, $B_2$ is a minimal spanning set and $|C_2| = p^{2n-r-t}.$ 
\end{proof}

Following the same process as in the above theorem, we can find the rank and the minimal spanning set of any cyclic code over the ring $R_k, k \geq 1$.
\begin{theorem} \label{rankk}
Let $n$ is not relatively prime to $p$. Let $C_k$ be a cyclic code of length $n$ over $R_k = \Z_p + u \Z_p + \cdots + u^{k-1} \Z_p, u^k = 0.$ We assume the constraints on the generator polynomials of $C_k$ as in Theorem \ref{gen-main}.
\begin{enumerate}[{\rm (1)}]
 \item If $ C_k = <g(x) + u p_1(x) + u^2 p_2(x) + \cdots + u^{k-1}p_{k-1}(x)> $ with deg $g(x) = r$, then $C_k$ is a free module with rank $n-r$ and a basis $B_1 = \{ g(x) + u p_1(x) + \cdots + u^{k-1}p_{k-1}(x), x(g(x) + u p_1(x) + \cdots + u^{k-1}p_{k-1}(x)), \dots, x^{n-r-1}(g(x) + u p_1(x) + \cdots + u^{k-1}p_{k-1}(x))\}.$
\item If $ C_k = < g(x) + u p_1(x) + u^2 p_{2}(x) + \cdots + u^{k-1} p_{k-1}(x), u a_1(x) + u^2 q_1(x) + \cdots + u^{k-1} q_{k-2}(x), u^2 a_2(x) + u^3 l_{1}(x) + \cdots + u^{k-1}l_{k-3}(x), \dots, u^{k-2} a_{k-2}(x) + u^{k-1} t_1(x), u^{k-1} a_{k-1}(x)> $ with deg $g(x) = r_1$, deg $a_1(x) =r_2$, deg $a_2(x) =r_3, \dots,$ deg $a_{k-1}(x) =r_k$, then $C_k$ has rank $n-r_k$ and a minimal spanning set  $B_2 = \{ g(x) + u p_1(x) + \cdots + u^{k-1}p_{k-1}(x), x(g(x) + u p_1(x) + \cdots + u^{k-1}p_{k-1}(x)), \dots, x^{n-r_1-1}(g(x) + u p_1(x) + \cdots + u^{k-1}p_{k-1}(x)), u a_1(x) + u^2 q_1(x) + \cdots + u^{k-1} q_{k-2}(x), x(u a_1(x) + u^2 q_1(x) + \cdots + u^{k-1} q_{k-2}(x)), \dots,\\ x^{r_1-r_2-1}(u a_1(x) + u^2 q_1(x) + \cdots + u^{k-1} q_{k-2}(x)), u^2 a_2(x) + u^3 l_{1}(x) + \cdots + u^{k-1}l_{k-3}(x), x(u^2 a_2(x) + u^3 l_{1}(x) + \cdots + u^{k-1}l_{k-3}(x)), \dots, x^{r_2-r_3-1}(u^2 a_2(x) + u^3 l_{1}(x) + \cdots + u^{k-1}l_{k-3}(x)), \dots, u^{k-1}a_{k-1}(x), xu^{k-1}a_{k-1}(x), \dots, x^{r_{k-1}-r_k-1}\\u^{k-1}a_{k-1}(x)\}.$
\item If $C_k = < g(x) + u p_1(x) + u^2 p_{2}(x) + \cdots + u^{k-1} p_{k-1}(x), u^{k-1}a_{k-1}(x)>$ with deg $g(x) = r$ and deg $a_{k-1}(x) = t$, then $C_k$ has rank $n-t$ and a minimal spanning set $B_3 = \{g(x) + u p_1(x) + u^2 p_{2}(x) + \cdots + u^{k-1} p_{k-1}(x), x(g(x) + u p_1(x) + u^2 p_{2}(x) + \cdots + u^{k-1} p_{k-1}(x)), \dots, x^{n-r-1}(g(x) + u p_1(x) + u^2 p_{2}(x) + \cdots + u^{k-1} p_{k-1}(x)), u^{k-1}a_{k-1}(x), xu^{k-1}a_{k-1}(x), \dots,\\ x^{r-t-1}u^{k-1}a_{k-1}(x)\}.$
\end{enumerate}
\end{theorem}
\begin{proof}
(1) The proof is same as in Theorem \ref{rank2}. Suppose $x^n - 1 = $
$$(g(x) + u p_1(x) + \cdots + u^{k-1}p_{k-1}(x)) (h(x) + u h_1(x) + \cdots + u^{k-1}h_{k-1}(x))$$ over $R_k$.  Suppose $\xn = (g(x) + u p(x))(h(x) + u h_1(x))$ over $R_2$. Let $c(x) \in C_k = <g(x) + u p_1(x) + u^2 p_2(x) + \cdots + u^{k-1}p_{k-1}(x)> $, then $c(x) = (g(x) + u p_1(x) + u^2 p_2(x) + \cdots + u^{k-1}p_{k-1}(x))f(x)$ for some polynomial $f(x)$. If deg $f(x) \leq n-r-1$, then $c(x)$ can be written as linear combinations of elements of $B_1$. Otherwise by the division algorithm there exist polynomials $q(x)$ and $r(x)$ such that
$$ f(x) = \left(\frac{\xn}{g(x) + u p_1(x) + \cdots + u^{k-1}p_{k-1}(x)}\right)q(x) + r(x)$$ where $r(x) = 0$ or deg $r(x) \leq n-r-1. $
This gives,
$$(g(x) + u p_1(x) + \cdots + u^{k-1}p_{k-1}(x))f(x) = (g(x) + u p_1(x)  + \cdots + u^{k-1}p_{k-1}(x))r(x).$$
Since deg $r(x) \leq n-r-1$, this shows that $B_1$ spans $C_k$. Now we only need to show that $B_1$ is linearly independent. Let $ g(x) = g_0 + g_1x + \cdots + g_rx^r$ and $p_1(x) = p_{1,0} + p_{1,1}x + \cdots + p_{1,l_1}x^{l_1},$ $ p_2(x) = p_{2,0} + p_{2,1}x + \cdots + p_{1,l_2}x^{l_2}, \dots,$ $p_{k-1}(x) = p_{k-1,0} + p_{k-1,1}x + \cdots + p_{k-1,l_{k-1}}x^{l_{k-1}},$ $g_0 \in \Z_p^{\times}, g_i, p_{j, i-1} \in \Z_p, i,j \geq 1$.\\ Suppose
$(g(x) + u p_1(x) + \cdots + u^{k-1}p_{k-1}(x))c_0 + x(g(x) + u p_1(x) + \cdots + u^{k-1}p_{k-1}(x))c_1 + \cdots + x^{n-r-1}(g(x) + u p_1(x) + \cdots + u^{k-1}p_{k-1}(x))c_{n-r-1} = 0.$  
By comparing the coefficients in the above equation, we get
$$ (g_0 + u p_{1,0} + \cdots + u^{k-1}p_{k-1,0}) c_0 = 0. ~\text{(constant coefficient)} $$
Since $(g_0 + u p_{1,0} + \cdots + u^{k-1}p_{k-1,0})$ is unit, we get $ c_0 = 0$. Thus,
$x(g(x) + u p_1(x) + \cdots + u^{k-1}p_{k-1}(x))c_1 + \cdots + x^{n-r-1}(g(x) + u p_1(x) + \cdots + u^{k-1}p_{k-1}(x))c_{n-r-1} = 0.$
Again comparing the coefficients, we get
$$ (g_0 + u p_{1,0} + \cdots + u^{k-1}p_{k-1,0}) c_1 = 0. ~\text{(coefficient of}~x). $$
As above, this gives $ c_1 = 0$. Continuing in this way we get that $c_i = 0 $ for all $ i = 0, 1. \dots, n-r-1$. Therefore, the set $B_1$ is linearly independent and hence a basis for $C_k$.\\
(2) If $C_k = <g(x) + u p_1(x) + \cdots + u^{k-1} p_{k-1}(x), u a_1(x) + u^2 q_1(x) + \cdots + u^{k-1} q_{k-2}(x), u^2 a_2(x) + u^3 l_{1}(x) + \cdots + u^{k-1}l_{k-3}(x), \dots, u^{k-2} a_{k-2}(x) + u^{k-1} t_1(x), \\u^{k-1} a_{k-1}(x)> $ with deg $(g(x) + u p_1(x) + \cdots + u^{k-1} p_{k-1}(x)) = r_1$, deg $(u a_1(x) + u^2 q_1(x) + \cdots + u^{k-1} q_{k-2}(x)) = r_2$, deg $(u^2 a_2(x) + u^3 l_{1}(x) + \cdots + u^{k-1}l_{k-3}(x)) = r_3, \dots$, and deg $(u^{k-1} a_{k-1}(x)) = r_k $. The lowest degree polynomial in $C_k$ is $u^{k-1} a_{k-1}(x)$. It is suffices to show that $B_2$ spans $ B =\{ g(x) + u p_1(x) + \cdots + u^{k-1}p_{k-1}(x), x(g(x) + u p_1(x) + \cdots + u^{k-1}p_{k-1}(x)), \dots, x^{n-r_1-1}(g(x) + u p_1(x) + \cdots + u^{k-1}p_{k-1}(x)), u a_1(x) + u^2 q_1(x) + \cdots + u^{k-1} q_{k-2}(x), x(u a_1(x) + u^2 q_1(x) + \cdots + u^{k-1} q_{k-2}(x)), \dots, x^{r_1-r_2-1}(u a_1(x) + u^2 q_1(x) + \cdots + u^{k-1} q_{k-2}(x)), u^2 a_2(x) + u^3 l_{1}(x) + \cdots + u^{k-1}l_{k-3}(x), x(u^2 a_2(x) + u^3 l_{1}(x) + \cdots + u^{k-1}l_{k-3}(x)), \dots, x^{r_2-r_3-1}\\(u^2 a_2(x) + u^3 l_{1}(x) + \cdots + u^{k-1}l_{k-3}(x)), \dots, u^{k-1}a_{k-1}(x), xu^{k-1}a_{k-1}(x), \dots, \\x^{n-r_k-1}u^{k-1}a_{k-1}(x)\}.$ As in the proof of part 2 of Theorem \ref{rank2}, it is suffices to show that $u^{k-1} x^{r_{k-1}-r_k} a_{k-1}(x) \in \text{span}(B_2)$. Let the leading coefficients of $ x^{r_{k-1}-r_k} a_{k-1}(x)$ be $a_0$ and of $g(x) + u p_1(x) + \cdots + u^{k-1}p_{k-1}(x)$ be $g_0$. There exists a constant $c_0 \in \Z_p$ such that $ a_0 = c_0 g_0$. Then we have
$$ u^{k-1} x^{r_{k-1}-r_k} a_{k-1}(x) =  u^{k-1} c_0(g(x) + u p_1(x) + \cdots + u^{k-1}p_{k-1}(x)) + u^{k-1} m(x), $$
where $u^{k-1}m(x)$ is a polynomial in $ C_k$ of degree less than $r_{k-1}$. Any polynomial in $C_k$ must have degree greater or equal to deg $(u^{k-1}a_{k-1}(x)) = r_k$. Hence, $ r_k \leq \text{deg}~m(x) < r_{k-1}$ and $u^{k-1}m(x)
 = \alpha_0 u^{k-1}a_{k-1}(x) + \alpha_1 x u^{k-1}a_{k-1}(x) + \cdots + \alpha_{r_{k-1}-r_k-1} x^{r_{k-1}-r_k-1} u^{k-1}a_{k-1}(x).$
Thus, $u^{k-1} x^{r_{k-1}-r_k} a_{k-1}(x) \in \text{span}(B_2).$ Hence $B_2$ is a generating set. As in (1), by comparing the coefficients we can see that $B_2$ is linearly independent. Therefore, $B_2$ is a minimal spanning set.\\
(3) This case is a special case of (2), so the proof is similar to case (2).
\end{proof}

\section{Minimum distance}
Let $n$ is not relatively prime to $p$. Let $C_2 = <g(x)+up(x), ua(x)>$ be a cyclic code of length $n$ over $R_2 = \Z_p + u \Z_p, u^2 = 0.$ We define $C_{2,u} = \{k(x) \in R_{2,n} : u k(x) \in C_2\}.$ It is easy to see that $C_{2,u}$ is a cyclic code over $\Z_p$. Let $ C_k $ be a cyclic code of length $n$ over $R_k = \Z_p + u \Z_p + \cdots + u^{k-1}\Z_p, u^k = 0.$ We define $ C_{k, u^{k-1}} = \{k(x) \in R_{k,n} : u^{k-1} k(x) \in C_k\}.$ Again it is easy to see that $C_{k, u^{k-1}}$ is a cyclic code over $\Z_p$.
\begin{theorem} \label{md1}
Let $n$ is not relatively prime to $p$. If $ C_k = < g(x) + u p_1(x) + u^2 p_{2}(x) + \cdots + u^{k-1} p_{k-1}(x), u a_1(x) + u^2 q_1(x) + \cdots + u^{k-1} q_{k-2}(x), u^2 a_2(x) + u^3 l_{1}(x) + \cdots + u^{k-1}l_{k-3}(x), \dots, u^{k-2} a_{k-2}(x) + u^{k-1} t_1(x), u^{k-1} a_{k-1}(x)> $ is a cyclic code of length $n$ over $R_k = \Z_p + u \Z_p + \cdots + u^{k-1}\Z_p, u^k = 0.$ Then $ C_{k, u^{k-1}} = <a_{k-1}(x)>$ and $w_{H}(C_{k}) = w_{H}(C_{k, u^{k-1}}).$ \end{theorem}
\begin{proof}
We have $u^{k-1} a_{k-1}(x) \in C_k$, thus $ <a_{k-1}(x)> \subseteq C_{k, u^{k-1}}.$ If $b(x) \in C_{k,u^{k-1}}$, then $u^{k-1}b(x) \in C_k$ and hence there exist polynomials $ b_1(x), \dots, b_k(x) \in \Z_p[X]$ such that $ u^{k-1}b(x) = b_1(x)u^{k-1}g(x) + b_2(x)u^{k-1}a_1(x) + b_2(x)u^{k-1}a_2(x)+ \cdots + b_k(x)u^{k-1}a_{k-1}(x).$
Since $a_{k-1}(x)|a_{k-2}(x)|\dots |a_{2}(x)|a_{1}(x)|g(x),$ we have $u^{k-1}b(x) = m(x)u^{k-1}a_{k-1}(x)$ for some polynomial $m(x) \in \Z_p[x]$. So, $C_{k, u^{k-1}} \in <a_{k-1}(x)> $, and hence $C_{k, u^{k-1}} = <a_{k-1}(x)>.$ Let $m(x) = m_0(x) + um_1(x) + \cdots + u^{k-1} m_{k-1}(x) \in C_k,$ where $m_0(x), m_1(x), \dots, m_{k-1}(x) \in \Z_p[x].$
We have $u^{k-1}m(x) = u^{k-1}m_0(x),$ $w_{H}(u^{k-1}m(x)) \leq w_{H}(m(x))$ and $u^{k-1}C_k$ is subcode of $C_k$ with $w_{H}(u^{k-1}C_k) \leq w_{H}(C_k)$. Therefore, it is sufficient to focus on the subcode $u^{k-1}C_k$ in order to prove the theorem. Since $u^{k-1}C_k = <u^{k-1}a_{k-1}(x)>$, we get $w_{H}(C_{k}) = w_{H}(C_{k, u^{k-1}}).$
\end{proof}
\begin{definition}
Let $ m = b_{l-1}p^{l-1} + b_{l-2}p^{l-2} + \cdots + b_1p + b_0$, $b_i \in \Z_p, 0 
\leq i \leq l-1$, be the $p$-adic expansion of $m$.
\begin{enumerate} [{\rm (1)}]
 \item If $ b_{l-i}  \neq 0$ for all $1  \leq i \leq q, q < l, $ and $ b_{l-i} = 0 $ for all $i, q+1 \leq i \leq l$, then $m$ is said to have a $p$-adic length $q$ zero expansion.
\item If $ b_{l-i}  \neq 0$ for all $1  \leq i \leq q, q < l, $ $b_{l-q-1} = 0$ and $ b_{l-i} \neq 0 $ for some $i, q+2 \leq i \leq l$, then $m$ is said to have  $p$-adic length $q$ non-zero expansion.
\item If $ b_{l-i}  \neq 0$ for $1  \leq i \leq l, $ then $m$ is said to have a $p$-adic length $l$  expansion or $p$-adic full expansion.
\end{enumerate}
\end{definition}
\begin{lemma} \label{md-lemma}
Let $C$ be a cyclic code over $R_k$ of length $p^l$ where $l$ is a positive integer. Let $C = <a(x)>$ where $a(x) = (x^{p^{l-1}} - 1)^bh(x)$, $ 1 \leq b < p$. If $h(x)$ generates a cyclic code of length $p^{l-1}$ and minimum distance $d$ then $d(C) = (b+1)d$.
\end{lemma}
\begin{proof}
For $ c \in C$, we have $ c = (x^{p^{l-1}} - 1)^bh(x)m(x)$ for some $ m(x) \in \frac{R_k[x]}{(x^{p^l}-1)}$. Since $h(x)$ generates a cyclic code of length $p^{l-1}$, we have $w(c) = w((x^{p^{l-1}} - 1)^bh(x)m(x)) = w(x^{p^{l-1}b}h(x)m(x)) + w(^bC_1x^{p^{l-1}(b-1)}h(x)m(x)) + \cdots + w(^bC_{b-1}\\x^{p^{l-1}}h(x)m(x)) + w(h(x)m(x))$. Thus, $ d(c) = (b + 1)d$.
\end{proof}

\begin{theorem} \label{md-thm}
Let $C_k$ be a cyclic code over $R_k$ of length $p^l$ where $l$ is a positive integer. Then,  $ C_k = < g(x) + u p_1(x) + u^2 p_{2}(x) + \cdots + u^{k-1} p_{k-1}(x), u a_1(x) + u^2 q_1(x) + \cdots + u^{k-1} q_{k-2}(x), u^2 a_2(x) + u^3 l_{1}(x) + \cdots + u^{k-1}l_{k-3}(x), \dots, u^{k-2} a_{k-2}(x) + u^{k-1} t_1(x), u^{k-1} a_{k-1}(x)> $ where $g(x) = (x-1)^{t_1}, a_1(x) = (x-1)^{t_2}, \dots, a_{k-1}(x) = (x-1)^{t_k}$. for some $ t_1 > t_2 > \cdots > t_k > 0.$ 
\begin{enumerate}[{\rm (1)}]
\item If $t_k \leq p^{l-1},$ then $d(C) = 2$. 
\item If $t_k > p^{l-1},$ let $t_k = b_{l-1}p^{l-1} + b_{l-2}p^{l-2} + \cdots + b_1p + b_0$ be the $p$-adic expansion of $t_k$ and $ a_{k-1}(x) = (x-1)^{t_k} = (x^{p^{l-1}} - 1)^{b_{l-1}}(x^{p^{l-2}} - 1)^{b_{l-2}} \cdots (x^{p^{1}} - 1)^{b_1}(x^{p^0} - 1)^{b_0}$.
\begin{enumerate}[{\rm ($a$)}]
 \item If $t_k$ has a $p$-adic length $q$ zero expansion or full expansion $(l=q)$. Then, $d(C_k) = (b_{l-1}+1)(b_{l-2}+1)\cdots(b_{l-q}+1).$
\item If $t_k$ has a $p$-adic length $q$ non-zero expansion. Then, $d(C_k) = 2(b_{l-1}+1)(b_{l-2}+1)\cdots(b_{l-q}+1)$
\end{enumerate}
\end{enumerate}
\end{theorem}
\begin{proof}
The first claim easily follows from Theorem \ref{gen-main}. From Theorem \ref{md1}, we see that $d(C_{k}) = d(u^{k-1}C_{k}) = d((x - 1)^{t_k})$. hence, we only need to determine the minimum weight of $u^{k-1}C_{k} = (x - 1)^{t_k}.$\\
(1) If $t_k \leq p^{l-1},$ then $(x - 1)^{t_k}(x - 1)^{p^{l-1}-t_k} = (x - 1)^{p^{l-1}} = (x^{p^{l-1}} - 1) \in C_k$. Thus, $d(C_k) = 2.$\\
(2) Let $ t_k > p^{l-1}$. (a) If $t_k$ has a $p$-adic length $q$ zero expansion, we have $t_k = b_{l-1} p^{l-1} + b_{l-2}p^{l-2} + \cdots + b_{l-q}p^{l-q}$, and $a_{k-1}(x) = (x - 1)^{t_k} = (x^{p^{l-1}}-1)^{b_{l-1}}(x^{p^{l-2}}-1)^{b_{l-2}}\cdots(x^{p^{l-q}}-1)^{b_{l-q}}.$ Let $ h(x) = (x^{p^{l-q}}-1)^{b_{l-q}}.$ Then $h(x)$ generates a cyclic code of length $ p^{l-q+1}$ and minimum distance $ (b_{l-q}+1)$. By Lemma \ref{md-lemma}, the subcode generated by $(x^{p^{l-q+1}}-1)^{b_{l-q+1}}h(x)$ has minimum distance $ (b_{l-q+1}+1) (b_{l-q}+1).$ By induction on $q$, we can see that the code generated by $a_{k-1}(x)$ has minimum distance $(b_{l-1}+1)(b_{l-2}+1)\cdots(b_{l-q}+1).$ Thus, $d(C_k) = (b_{l-1}+1)(b_{l-2}+1)\cdots(b_{l-q}+1).$\\
(b) If $t_k$ has a $p$-adic length $q$ non-zero expansion, we have $t_k = b_{l-1} p^{l-1} + b_{l-2}p^{l-2} + \cdots + b_{1}p + b_0, b_{l-q-1} = 0.$ Let $ r = b_{l-q-2}p^{l-q-2}+ b_{l-q-3}p^{l-q-3}+ \cdots + b_1p + b_0$ and $ h(x) = (x-1)^r = (x^{p^{l-q-2}}-1)^{b_{l-q-2}}(x^{p^{l-q-3}}-1)^{b_{l-q-3}}\cdots(x^{p^{1}}-1)^{b_{1}}(x^{p^{0}}-1)^{b_{0}}.$ Since $ r < p^{l-q-1}$, we have $ p^{l-q-1} = r+j$ for some non-zero $j$. Thus, $ (x-1)^{p^{l-q-1}-j}h(x) = (x^{p^{l-q-1}} - 1) \in C_k.$ Hence, the subcode generated by $h(x)$ has minimum distance 2. By Lemma \ref{md-lemma}, the subcode generated by $ (x^{p^{l-q}} - 1)^{b_{l-q}}h(x)$ has minimum distance $2(b_{l-q}+1)$. By induction on $q$, we can see that the code generated by $a_{k-1}(x)$ has minimum distance $2(b_{l-1}+1)(b_{l-2}+1)\cdots(b_{l-q}+1).$ Thus, $d(C_k) = 2(b_{l-1}+1)(b_{l-2}+1)\cdots(b_{l-q}+1).$\\
\end{proof}

\section{Examples}
\begin{example}
Cyclic codes of length $5$ over $R_4 = \Z_3 + u \Z_3 + u^2 \Z_3 + u^3 \Z_3, u^4 = 0$: We have
$$ x^5-1 = (x-1)(x^4 + x^3 + x^2 + x +1) = g_1g_2 ~\text{over}~ R_4.$$
The non-zero cyclic codes of length $5$ over $R_4$ with generator polynomial are given in Table 1.
\end{example}
\begin{center}
 {\bf Table 1.} Cyclic codes of length 5 over $R_4$.
\begin{tabular}{| c | }
\hline
 Non-zero generator polynomials\\
\hline
$<1>, <g_1>, <g_2>$\\
\hline
$<u>, <ug_1>, <ug_2>$\\
\hline
$<u^2>, <u^2g_1>, <u^2g_2>$\\
\hline
$<u^3>, <u^3g_1>, <u^3g_2>$\\
\hline
$<g_1, u>, <g_2, u>, <g_1, u^2>, <g_2, u^2>, <g_1, u^3>, <g_2, u^3>$\\
\hline
$<ug_1, u^2>, <ug_2, u^2>$\\
\hline
$<u^2g_1, u^3>, <u^2g_2, u^3>.$\\
\hline
\end{tabular}
\end{center}

\bibliographystyle{plain}
\bibliography{ref}

\end{document}